%% file: games.tex
\documentclass{llncs}

\usepackage[utf8x]{inputenc}

\usepackage{amsmath}

\usepackage{amsthm}
\usepackage{amssymb}
\usepackage{breqn}
\usepackage{array}
\usepackage{enumerate}
\usepackage{bookmark}
\usepackage{hyperref}
\hypersetup{
  pdftitle = {Provable Advantage for Quantum Strategies in Random Symmetric XOR Games},
  pdfauthor = {Andris Ambainis, Janis Iraids},
  pdfkeywords = {quantum, random, nonlocal, games},
  pdfstartview=FitH,
  unicode=true
}

\begin{document}

\title{Provable Advantage for Quantum Strategies in Random Symmetric XOR Games}
\author{\texorpdfstring{Andris Ambainis, Jānis Iraids}{Andris Ambainis, Jānis Iraids}}
\titlerunning{Provable Advantage for Quantum Strategies in Random Symmetric XOR Games}
\institute{Faculty of Computing, University of Latvia,\\ Raiņa bulvāris 19, Riga, LV-1586, Latvia}

\maketitle

\bookmarksetup{startatroot}

\begin{abstract}
Non-local games are widely studied as a model to investigate the properties of quantum mechanics as opposed to classical mechanics. 
In this paper, we consider a subset of non-local games: symmetric XOR games of $n$ players with 0-1 valued questions. 
For this class of games, each player receives an input bit and responds with an output bit without communicating to the other players. 
The winning condition only depends on XOR of output bits and is constant w.r.t. permutation of players.

We prove that for almost any $n$-player symmetric XOR game the entangled value of the game is $\Theta\left (\frac{\sqrt{\ln{n}}}{n^{1/4}}\right )$ adapting an old result by Salem and Zygmund on the asymptotics of random trigonometric polynomials. 
Consequently, we show that the classical-quantum gap is $\Theta (\sqrt{\ln{n}})$ for almost any symmetric XOR game.
\end{abstract}

\section{Introduction}

Non-local games provide a simple way to test the difference between quantum mechanics 
and the classical world. A prototypical example of a non-local game is the CHSH game
\cite{CHTW04} (based on the CHSH inequality of \cite{CHSH69}). In the CHSH game, we have
two players who cannot communicate between themselves but may share common random bits or a bipartite quantum state (which has been exchanged before the beginning of the game). 
A referee sends one uniformly random bit $a\in\{0, 1\}$ to the $1^{\rm st}$ player
and another bit to the $2^{\rm nd}$ player. Players respond by sending one-bit 
answers $x, y\in\{0, 1\}$. They win in the following 2 cases:
\begin{enumerate}[(a)]
\item
If at least one of $a, b$ is equal to 0, players win if they produce 
$x, y$ such that $x=y$;
\item
If $a = b = 1$, players win if they produce $x, y$ such that $x\neq y$;
\end{enumerate}
Classically, CHSH game can be won with probability at most 0.75. In contrast, if
players use an entangled quantum state, they can win the game with probability
$\frac{1}{2}+\frac{1}{2\sqrt{2}}=0.85...$.

Other non-local games can be defined by changing the number of players, the number of
possible questions and answers and the winning condition. Many non-local games have been
studied and, in many cases, strategies that use an entangled quantum state outperform any classical strategy.

Recently \cite{ABB+12}, it has been shown that, for a large class of non-local games, quantum 
strategies are better than any classical strategy for almost all games in this class. 
Namely, \cite{ABB+12} considered 2-player games in which the questions $a, b$ are taken from 
the set $\{1, 2, \ldots, n\}$ and the winning condition is either $x=y$ or $x\neq y$,
depending on $a, b$. (Games with a winning condition of such form are called XOR games.)
\cite{ABB+12} showed that, for $1-o(1)$ fraction of all such games, the entangled value of the 
game is at least 1.2... times its classical value.

Then \cite{AIKV13}, it was discovered that a similar effect might hold for another class of
games: $n$-player symmetric XOR games with binary questions. Namely, \cite{AIKV13} showed a gap 
between entangled and classical values of order $\Omega(\sqrt{\log n})$ - assuming that a non-rigorous argument about the entangled value is correct.

In this paper, we make this gap rigorous, by proving upper and lower bounds on the entangled
value of a random game in this class. We show that, with a high probability, the entangled
value is equal to $\Theta(\frac{\sqrt{\log n}}{n^{1/4}})$. The quantum-vs-classical gap
of $\Theta(\sqrt{\log n})$ follows by combining this with the fact that the classical 
value is of the order $\Theta(\frac{1}{n^{1/4}})$ (shown in \cite{AIKV13}).

To prove this result, we use an expression for the entangled value of a symmetric $n$-player
XOR game with entangled answers from \cite{AKNR10}. This expression reduces finding
the entangled value to maximizing the absolute value of a polynomial in one complex variable.
If conditions for the XOR game are chosen at random, this expression reduces to 
random trigonometric polynomials studied in \cite{SZ54}. 

Although maxima of random trigonometric polynomials have been studied in \cite{SZ54},
they have been studied under different conditions. For this reason, we cannot
apply the results from \cite{SZ54} directly. Instead, we adapt the ideas from \cite{SZ54}
to prove a bound on maxima of random trigonometric polynomials that would be
applicable in our setting.

\section{Definitions}

A non-local game with $n$ players proceeds as follows:
\begin{enumerate}[1)]
\item Players are separated so that they cannot communicate -- hence the name \emph{non-local},
\item The players receive inputs $x_1, x_2, \ldots, x_n \in I$ where $I$ is the set of possible inputs. $i$-th player receives $x_i$,
\item The players respond with outputs $y_1, y_2, \ldots, y_n \in O$ where $O$ is the set of possible outputs.
\item The winning condition $P(x_1, \ldots, x_n, y_1, \ldots, y_n)$ is consulted to determine whether the players win or lose. The condition is known to everyone at the start of the game.
\end{enumerate}

The players are informed of the rules of the game and they can agree upon a strategy and exchange other information.
In the classical case players may only use shared randomness. In the quantum case they can use an entangled quantum state which is distributed to the players before the start of the game.

We will restrict ourselves to the case when $I=O=\{0, 1\}$ and the vector of inputs $(x_1, \ldots, x_n)$ is chosen uniformly at random.
In an \emph{XOR} game, the winning condition $P(x_1, \ldots, x_n, y_1, \ldots, y_n)$ depends only on $x_1, \ldots, x_n$ and the 
parity of the output bits $\oplus_{j=1}^n{y_j}$. A game is \emph{symmetric} if the winning condition does not change if $x_1, \ldots, x_n$ are permuted. 

The winning conditions of a symmetric XOR game can be described by a list of $n+1$ bits: $G=(G_0, G_1, \ldots, G_n)$, where the players win if and only if $G_i=\oplus_{j=1}^n{y_j}$ when $\sum_{j=1}^n{x_j}=i$. 

The \emph{entangled value} of the game $Val_Q(G)$ is the probability of winning minus the probability of losing in the conditions that the players can use a shared quantum-physical system. In this paper, we study the value of symmetric XOR games when the winning condition $G$ is chosen randomly from the uniform distribution of all $(n+1)$-bit lists.
We use the following lemma (which follows from a more general result by Werner and Wolf for non-symmetric XOR games \cite{WW01}):

\begin{lemma}[See \cite{AKNR10}]
\label{lemma:val}
The entangled value of a symmetric XOR game \cite{AKNR10} is
\begin{equation}
\label{eq:val}
  Val_Q(G)=\max_{|\lambda|=1}{\left|\sum_{j=0}^{n}{(-1)^{G_j}p_j \lambda^j}\right|}
\end{equation}
where $p_j$ is the probability that players are given an input vector $(x_1, \ldots, x_n)$ with $j$ variables $x_i=1$.
\end{lemma}

In our case, since $(x_1, \ldots, x_n)$ is uniformly random, 
we have $p_j=\frac{\binom{n}{j}}{2^n}$. 

In the following sections we introduce additional notation to keep the proofs more concise as well as to keep in line with the original proofs in \cite{SZ54}:

The \emph{Rademacher system} is a set of functions $\{\varphi_m(t)\}$ for $m = 1, \ldots, n$ over $0 \leq t \leq 1$ such that $\varphi_m(t)=(-1)^k$, where $k$ is the $m$-th digit after the binary point in the binary expansion of $t$. Rademacher system will turn out to be a convenient way to state that $\{G_j\}$ are random variables that follow a uniform distribution: if $t$ is chosen randomly from a uniform distribution on $0\leq t\leq 1$, then $\{\varphi_m(t)\}_{m=1}^{n+1}$ generates a uniformly random element from $\{+1,-1\}^{n+1}$. That in turn corresponds to coefficients $(-1)^{G_j}$ in eq. (\ref{eq:val}) being picked randomly.

Furthermore, we define
\[ r_m = \binom{n}{m} \quad (n \text{ will be clear from context}), \]
\[ R_n = \sum_{m=0}^n{r_m^2}, \]
\[ T_n = \sum_{m=0}^n{r_m^4}, \]
\[ P_n(x,t) = \sum_{m=0}^n{r_m \varphi_{m+1}(t) \cos{m x}}, \]
\[ M_n(t) = \max_{ 0 \leq x < 2\pi}{|P_n(x,t)|}.\]

\section{Main Result}

By adapting the work of Salem and Zygmund \cite{SZ54} on the asymptotics of random trigonometric polynomials, we show

\begin{theorem}
\label{thm:lower}
\[\lim_{n\rightarrow \infty}{\Pr[M_n(t) \geq C_1 \sqrt{R_n \ln{n}}]}=1\]
\end{theorem}

\begin{theorem}
\label{thm:upper}
\[\lim_{n\rightarrow \infty}{\Pr[M_n(t) \leq C_2 \sqrt{R_n \ln{n}}]}=1\]
\end{theorem}
Our proof yields $C_1=\frac{1}{4\sqrt{3}}$ and $C_2=2$.

We will now show how these two theorems lead to an asymptotic bound for the entangled value of a random game.

\begin{corollary}
For almost all $n$-player symmetric quantum XOR games the value of the game is asymptotically $\frac{\sqrt{\ln{n}}}{n^{1/4}}$.
\end{corollary}

\begin{proof}
From Lemma \ref{lemma:val},
\[Val_Q(G) \geq \max_{|\lambda|=1}{\left|\mathfrak{Re}\left(\sum_{j=0}^{n}{\frac{(-1)^{G_j}\binom{n}{j}\lambda^j}{2^n}}\right)\right|} = \max_{\alpha \in [0;2\pi]}{\left|\sum_{j=0}^{n}{\frac{(-1)^{G_j}\binom{n}{j}\cos{j\alpha}}{2^n}}\right|},\]
and
\[\begin{split}Val_Q(G) \leq \max_{|\lambda|=1}{\left|\mathfrak{Re}\left(\sum_{j=0}^{n}{\frac{(-1)^{G_j}\binom{n}{j}\lambda^j}{2^n}}\right)\right|} + \max_{|\lambda|=1}{\left|\mathfrak{Im}\left(\sum_{j=0}^{n}{\frac{(-1)^{G_j}\binom{n}{j}\lambda^j}{2^n}}\right)\right|} =\\
= \max_{\alpha \in [0;2\pi]}{\left|\sum_{j=0}^{n}{\frac{(-1)^{G_j}\binom{n}{j}\cos{j\alpha}}{2^n}}\right|} + \max_{\alpha \in [0;2\pi]}{\left|\sum_{j=0}^{n}{\frac{(-1)^{G_j}\binom{n}{j}\sin{j\alpha}}{2^n}}\right|}
\end{split}
\]
For a random game $\{(-1)^{G_j}\}$ follow the same distribution as $\left\{\varphi_{j+1}(t)\right\}$ for $t$ uniformly distributed from interval $[0;1]$. Therefore Theorem \ref{thm:lower} and Theorem \ref{thm:upper} apply. Note that Theorem \ref{thm:upper} is true for cosines as well as sines since we only use that $\cos^2{x} \leq 1$, and so
\begin{equation}
\label{eq:bounds}
\lim_{n\rightarrow \infty}{\Pr\left[C_1\frac{\sqrt{R_n \ln{n}}}{2^n} \leq Val_Q(G) \leq 2C_2 \frac{\sqrt{R_n \ln{n}}}{2^n} \right]} = 1
\end{equation}
Finally, 
\[\frac{\sqrt{R_n \ln{n}}}{2^n} = \frac{\sqrt{\binom{2n}{n} \ln{n}}}{2^n} \sim \frac{\sqrt{\frac{4^n}{\sqrt{\pi n}} \ln{n}}}{2^n} = \sqrt{\frac{\ln{n}}{\sqrt{\pi n}}}\]
\end{proof}

\section{Proof of Upper and Lower Bounds}

We now proceed to prove theorems \ref{thm:lower} and \ref{thm:upper}.
Our proof is based on an old result by Salem and Zygmund \cite{SZ54}, 
in which they prove bounds on the asymptotics of random trigonometric polynomials
in a different setting (in which the coefficients $r_m$ are not allowed to depend on $n$). 

Due to the difference in the two settings, we cannot immediately apply the results
from \cite{SZ54}. Instead, we prove corresponding theorems for our setting, re-using the 
parts of proof from \cite{SZ54} which also work in our case and replacing other parts with
different arguments.

\begin{lemma}[From \cite{SZ54}]
\label{lemma:1}
Let $f_n(t)=\sum_{m=0}^n{c_m \varphi_{m+1}(t)}$, where $\{\varphi_{m+1}(t)\}$ is the Rademacher system and $c_m$ are real constants. Let $C_n=\sum_{m=0}^n{c_m^2},D_n=\sum_{m=0}^n{c_m^4}$ and let $\lambda$ be any real number. Then
\[e^{\frac{1}{2}\lambda^2C_n - \lambda^4 D_n} \leq \int_0^1{e^{\lambda f_n(t)}\,\mathrm{d}t} \leq e^{\frac{1}{2}\lambda^2 C_n}.\]
\end{lemma}

\begin{lemma}[From \cite{SZ54}]
\label{lemma:2}
Let $g(x,y)$, $a \leq x \leq b$, $c \leq y \leq d$, be a bounded real function. Suppose that
\[|g(x,y)|\leq A, \frac{\int_c^d{\int_a^b{g^2(x,y)\,\mathrm{d}x}\,\mathrm{d}y}}{(b-a)(d-c)}=B.\]
Then, for any positive number $\mu$, 
\[\frac{\int_c^d{\int_a^b{e^{\mu g(x,y)}\,\mathrm{d}x}\,\mathrm{d}y}}{(b-a)(d-c)}\leq 1 + \mu \sqrt{B} + \frac{B}{A^2}e^{\mu A}.\]
Furthermore, when $\int_c^d{\int_a^b{g(x,y)\,\mathrm{d}x}\,\mathrm{d}y}=0$,
\begin{equation}
\label{eq:exp}
\frac{\int_c^d{\int_a^b{e^{\mu g(x,y)}\,\mathrm{d}x}\,\mathrm{d}y}}{(b-a)(d-c)}\leq 1 + \frac{B}{A^2}e^{\mu A}.
\end{equation}
\end{lemma}

\begin{lemma}[From \cite{SZ54}]
\label{lemma:3}
Let $x$ be real and $P(x)=\sum_{m=0}^n{\alpha_m \cos mx + \beta_m \sin mx}$ be a tri\-go\-nometric polynomial of order $n$, with real or imaginary coefficients. Let $M$ denote the maximum of $|P(x)|$ and let $\theta$ be a positive number less than 1. Then there exists an interval of length not less than $\frac{1-\theta}{n}$ in which $|P(x)|\geq \theta M$.
\end{lemma}

\begin{lemma}[From \cite{SZ54}]
\label{lemma:4}
Let $\varphi(x)\geq 0$, and suppose that
\[\int_0^1{\varphi(x)\,\mathrm{d}x}\geq A > 0, \int_0^1{\varphi^2(x)\,\mathrm{d}x}\leq B\]
(clearly, $A^2\leq B$). Let $0<\delta<1$. Then
\[\Pr\left[\varphi(x)\geq \delta A \ |\  0 \leq x \leq 1\right] \geq (1-\delta )^2\frac{A^2}{B}.\]
\end{lemma}

\begin{lemma}
\label{lemma:5}
\[\frac{\sum_{i=0}^n{\binom{n}{i}^4}}{\left(\sum_{i=0}^n{\binom{n}{i}^2}\right)^2}\leq \frac{4}{3} n^{-\frac{1}{2}}\]
\end{lemma}
\begin{proof}
If $n$ is even:
\[\begin{split}&\frac{\sum_{i=0}^n{\binom{n}{i}^4}}{\left(\sum_{i=0}^n{\binom{n}{i}^2}\right)^2} \leq \frac{\sum_{i=0}^n{\binom{n}{i}^2}\binom{n}{n/2}^2}{\left(\sum_{i=0}^n{\binom{n}{i}^2}\right)^2} = \frac{\binom{n}{n/2}^2}{\binom{2n}{n}} \leq \\
&\leq \frac{\left(\frac{2^n}{\sqrt{3\frac{n}{2}+1}}\right)^2}{\frac{4^n}{\sqrt{4n}}} \leq \frac{\sqrt{4n}}{3\frac{n}{2}+1} \leq \frac{4}{3} n^{-\frac{1}{2}}\end{split}\]

If $n$ is odd:
\[\begin{split}&\frac{\sum_{i=0}^n{\binom{n}{i}^4}}{\left(\sum_{i=0}^n{\binom{n}{i}^2}\right)^2} \leq \frac{\sum_{i=0}^n{\binom{n}{i}^2}\binom{n}{\lfloor n/2 \rfloor}^2}{\left(\sum_{i=0}^n{\binom{n}{i}^2}\right)^2} = \frac{\left(\frac{\binom{n+1}{\frac{n+1}{2}}}{2}\right)^2}{\binom{2n}{n}} \leq \\
&\leq \frac{\left(\frac{2^{n+1}}{2\sqrt{3\frac{n+1}{2}+1}}\right)^2}{\frac{4^n}{\sqrt{4n}}} \leq \frac{\sqrt{4n}}{3\frac{n+1}{2}+1} \leq \frac{4}{3} n^{-\frac{1}{2}}\end{split}\]
\end{proof}

\begin{proof}[Proof of Theorem \ref{thm:lower}]
Set $I_n(t)=\frac{1}{2\pi}\int_0^{2\pi}{e^{\lambda P_n(x,t)}\,\mathrm{d}x}$. We proceed to give an upper bound for for $\int_0^1{I_n(t)\,\mathrm{d}t}$ and lower bound for $\int_0^1{I_n^2(t)\,\mathrm{d}t}$ using Lemma \ref{lemma:1}. Then we will plug in these bounds in Lemma \ref{lemma:4} for $\varphi = I_n$.

First, the lower bound clause of Lemma \ref{lemma:1} applied to $I_n(t)$ gives for any real $\lambda $ (we will assign its value later, at our convenience),
\[\begin{split}&\int_0^1{I_n(t)\,\mathrm{d}t}=\int_0^1{\left(\frac{1}{2\pi}\int_0^{2\pi}{e^{\lambda P_n(x,t)}\,\mathrm{d}x}\right)\,\mathrm{d}t} = \frac{1}{2\pi}\int_0^{2\pi}{\int_0^1{e^{\lambda P_n(x,t)}\,\mathrm{d}t}\,\mathrm{d}x}\geq \\
\geq &\frac{1}{2\pi}\int_0^{2\pi}{e^{\frac{1}{2}\lambda^2\sum_{m=0}^n{(r_m \cos{mx})^2} - \lambda^4\sum_{m=0}^n{(r_m \cos{mx})^4}}\,\mathrm{d}x} \geq \\
\geq &\frac{1}{2\pi}\int_0^{2\pi}{e^{\frac{1}{2}\lambda^2\sum_{m=0}^n{(r_m \cos{mx})^2} - \lambda^4T_n}\,\mathrm{d}x} = \\
&= \left(e^{\frac{1}{4}\lambda^2R_n-\lambda^4T_n}\right)\cdot \frac{1}{2\pi}\int_0^{2\pi}{e^{\frac{1}{2}\lambda^2\sum_{m=0}^n{(r_m \cos{mx})^2-\frac{r_m^2}{2}}}\,\mathrm{d}x} = \\
&= \left(e^{\frac{1}{4}\lambda^2R_n-\lambda^4T_n}\right)\cdot \frac{1}{2\pi}\int_0^{2\pi}{e^{\frac{1}{4}\lambda^2\sum_{m=0}^n{(r_m^2 \cos{2mx})}}\,\mathrm{d}x} > \\
&> \left(e^{\frac{1}{4}\lambda^2R_n-\lambda^4T_n}\right)\cdot \frac{1}{2\pi}\int_0^{2\pi}{\left(1+\frac{1}{4}\lambda^2\sum_{m=0}^n{(r_m^2 \cos{2mx})}\right)\,\mathrm{d}x} \geq \\
&\geq \left(e^{\frac{1}{4}\lambda^2R_n-\lambda^4T_n}\right)
\end{split}
\]
The second step is to establish an upper bound for $\int_0^1{I_n^2(t)\,\mathrm{d}t}$. We start out in a similar fashion, by applying Lemma \ref{lemma:1}:
\[\begin{split}&\int_0^1{I_n^2(t)\,\mathrm{d}t} = \frac{1}{(2\pi)^2}\int_0^{2\pi}{\int_0^{2\pi}{\int_0^1{e^{\lambda (P_n(x,t)+ P_n(y,t))}\,\mathrm{d}t}\,\mathrm{d}x}\,\mathrm{d}y} \leq \\
&\leq \frac{1}{(2\pi)^2}\int_0^{2\pi}{\int_0^{2\pi}{e^{\frac{1}{2}\lambda^2\sum_{m=0}^n{r_m^2(\cos{m x}+\cos{m y})^2}}\,\mathrm{d}x}\,\mathrm{d}y} = \\
&= e^{\frac{1}{2}\lambda^2(R_n+r_0^2)}\cdot \frac{1}{(2\pi)^2}\int_0^{2\pi}{\int_0^{2\pi}{e^{\frac{1}{2}\lambda^2 S_n(x,y)}\,\mathrm{d}x}\,\mathrm{d}y}
\end{split}\]
where 
\[S_n(x,y)=\sum_{m=1}^n{\left(\frac{1}{2}r_m^2 \cos{2m x} + \frac{1}{2}r_m^2 \cos{2m y} + 2r_m^2\cos{m x}\cos{m y}\right)}.\]
One can verify that
\begin{enumerate}[a)]
  \item \[\int_0^{2\pi}{\int_0^{2\pi}{S_n(x,y)\,\mathrm{d}x}\,\mathrm{d}y} = 0, \]
  \item
\[\begin{split}
&\frac{1}{(2\pi)^2} \int_0^{2\pi}{\int_0^{2\pi}{S_n(x,y)^2\,\mathrm{d}x}\,\mathrm{d}y} =\\
&= \frac{1}{2\pi} \sum_{m=1}^n {\int_0^{2\pi}{\left(\frac{1}{2}r_m^2 \cos{2m x}\right)^2\,\mathrm{d}x}} + \\
&+ \frac{1}{2\pi} \sum_{m=1}^n {\int_0^{2\pi}{\left(\frac{1}{2}r_m^2 \cos{2m y}\right)^2\,\mathrm{d}y}} + \\
&+ \frac{1}{(2\pi)^2} \sum_{m=1}^n {\int_0^{2\pi}{\int_0^{2\pi}{\left(2r_m^2\cos{m x}\cos{m y}\right)^2\,\mathrm{d}x}\,\mathrm{d}y}} = \\
&= \frac{5}{4}T_n
\end{split}\]
  \item \[\left|S_n(x,y)\right| \leq 3 R_n\]
\end{enumerate}
We apply eq. \ref{eq:exp} from Lemma \ref{lemma:2} with function $g = S_n$, $\mu = \frac{1}{2}\lambda^2$, $A = 3R_n$ and $B=\frac{5}{4}T_n$. We get
\[\begin{split}\frac{1}{(2\pi)^2}\int_0^{2\pi}{\int_0^{2\pi}{e^{\frac{1}{2}\lambda^2 S_n(x,y)}\,\mathrm{d}x}\,\mathrm{d}y} &\leq 1+\frac{\frac{5}{4}T_n}{9 R_n^2}e^{\frac{3}{2}\lambda^2R_n} \leq \\
&\leq 1+\frac{T_n}{R_n^2}e^{\frac{3}{2}\lambda^2R_n}
\end{split}
\]
And by Lemma \ref{lemma:5},
\[1+\frac{T_n}{R_n^2}e^{\frac{3}{2}\lambda^2R_n} \leq 1+\frac{4}{3}n^{-\frac{1}{2}}e^{\frac{3}{2}\lambda^2R_n}\]

So far we have established the two prerequisites for Lemma \ref{lemma:4}:
\begin{enumerate}[1)]
\item \[\int_0^1{I_n(t)\,\mathrm{d}t} > e^{\frac{1}{4}\lambda^2R_n-\lambda^4T_n},\]
\item \[\int_0^1{I_n^2(t)\,\mathrm{d}t} \leq e^{\frac{1}{2}\lambda^2(R_n+r_0^2)}\left(1+\frac{4}{3}n^{-\frac{1}{2}}e^{\frac{3}{2}\lambda^2R_n}\right).\]
\end{enumerate}

The third step is to apply Lemma \ref{lemma:4} with $\varphi = I_n$, $A = e^{\frac{1}{4}\lambda^2R_n-\lambda^4T_n}$, $B = e^{\frac{1}{2}\lambda^2(R_n+r_0^2)}\left(1+\frac{4}{3}n^{-\frac{1}{2}}e^{\frac{3}{2}\lambda^2R_n}\right)$ and $\delta = n^{-\eta}$. This results in
\[\begin{split}\Pr[I_n(t) \geq n^{-\eta}e^{\frac{1}{4}\lambda^2R_n-\lambda^4T_n} ] &\geq (1-n^{-\eta})^2\frac{e^{\frac{1}{2}\lambda^2R_n-2\lambda^4T_n}}{e^{\frac{1}{2}\lambda^2(R_n+r_0^2)}\left(1+\frac{4}{3}n^{-\frac{1}{2}}e^{\frac{3}{2}\lambda^2R_n}\right)} \geq \\
&\geq (1-n^{-\eta})^2e^{-2\lambda^4T_n-\frac{1}{2}\lambda^2r_0^2}\left(1-\frac{4}{3}n^{-\frac{1}{2}}e^{\frac{3}{2}\lambda^2R_n}\right)
\end{split}
  \]
Finally we show that for suitably chosen $\lambda$, $\theta$ and $\eta$ the claim follows. Set $\lambda = \theta \sqrt{\frac{\ln{n}}{R_n}}$ having $\theta$ such that $2\sqrt{\eta} < \theta < \sqrt{\frac{1}{3}}$. We deal with the two claims separately:
\begin{claim}
\[I_n(t) \geq n^{-\eta}e^{\frac{1}{4}\lambda^2R_n-\lambda^4T_n} \implies M_n(t) \geq C_1\sqrt{R_n \ln{n}}\]
\end{claim}
\begin{proof}
Note that
\[e^{\lambda M_n(t)} \geq I_n(t) \geq e^{\frac{1}{4}\lambda^2R_n-\lambda^4T_n-\eta \ln{n}}\]
Thus
\[\begin{split}M_n(t) &\geq \frac{1}{4}\lambda R_n-\lambda^3T_n-\frac{\eta}{\lambda} \ln{n} =\\
&= \frac{\theta}{4}\sqrt{R_n \ln{n}}-\theta^3\sqrt{R_n \ln{n}}\ln{n}\frac{T_n}{R_n^2} - \frac{\eta}{\theta}\sqrt{R_n \ln{n}} =\\
&= \sqrt{R_n \ln{n}}\left(\frac{\theta}{4}-\theta^3\frac{4\ln{n}}{3\sqrt{n}}-\frac{\eta}{\theta}\right) \rightarrow \sqrt{R_n \ln{n}}\left(\frac{\theta}{4}-\frac{\eta}{\theta}\right)
\end{split}
\]
But $ \frac{\theta}{4}-\frac{\eta}{\theta} = \text{constant} > 0$. We can choose $\theta$ arbitrarily close to $\sqrt{\frac{1}{3}}$ and $\eta$ arbitrarily close to 0 to obtain $C_1=\frac{1}{4\sqrt{3}}$.
\end{proof}
\begin{claim}
\[\lim_{n\rightarrow \infty}{(1-n^{-\eta})^2e^{-2\lambda^4T_n-\frac{1}{2}\lambda^2r_0^2}\left(1-\frac{4}{3}n^{-\frac{1}{2}}e^{\frac{3}{2}\lambda^2R_n}\right)} = 1\]
\end{claim}
\begin{proof}
Since $\eta$ is positive, $n^{-\eta}\rightarrow 0$.
\[e^{-2\lambda^4T_n-\frac{1}{2}\lambda^2r_0^2} = e^{-2\theta^4(\ln{n})^2\frac{T_n}{R_n^2}-\frac{1}{2}\theta^2r_0^2\frac{\ln{n}}{R_n}} \geq e^{-\frac{8}{3\sqrt{n}}\theta^4(\ln{n})^2-\frac{1}{2}\theta^2r_0^2\frac{\ln{n}}{R_n}}\rightarrow e^0 = 1\]
\[\frac{4}{3}n^{-\frac{1}{2}}e^{\frac{3}{2}\lambda^2R_n} = \frac{4}{3}n^{-\frac{1}{2}}e^{\frac{3}{2}\theta^2\ln{n}} = \frac{4}{3}n^{\frac{3\theta^2-1}{2}}\rightarrow 0\]
\end{proof}

\end{proof}

\begin{proof}[Proof of Theorem \ref{thm:upper}]
We will examine $\int_0^1{\int_0^{2\pi}{e^{\lambda |P_n(x,t)|}\,\mathrm{d}x}\,\mathrm{d}t}$. By Lemma \ref{lemma:3} there exists $0 < \theta < 1$ such that:
  \[\begin{split}
  &\int_0^1{\int_0^{2\pi}{e^{\lambda |P_n(x,t)|}\,\mathrm{d}x}\,\mathrm{d}t}\geq \\
  &\geq \int_0^1{\frac{1-\theta}{n}e^{\theta \lambda M_n(t)}\,\mathrm{d}t}
  \end{split}\]
On the other hand, by Lemma \ref{lemma:1} we obtain:
  \[\begin{split}
  &\int_0^1{\int_0^{2\pi}{e^{\lambda |P_n(x,t)|}\,\mathrm{d}x}\,\mathrm{d}t} =\\
  &= \int_0^{2\pi}{\int_0^1{e^{\lambda |P_n(x,t)|}\,\mathrm{d}t}\,\mathrm{d}x} \leq\\
  &\leq \int_0^{2\pi}{\int_0^1{e^{\lambda P_n(x,t)}+e^{-\lambda P_n(x,t)}\,\mathrm{d}t}\,\mathrm{d}x} \leq \\
  &\leq \int_0^{2\pi}{\int_0^1{2 e^{\frac{1}{2}\lambda^2 \sum_{m=0}^n{r_m^2 \cos^2{mx}}}\,\mathrm{d}t}\,\mathrm{d}x} \leq \\
  &\leq \int_0^{2\pi}{\int_0^1{2 e^{\frac{1}{2}\lambda^2 R_n}\,\mathrm{d}t}\,\mathrm{d}x} = \\
  &= 4\pi e^{\frac{1}{2}\lambda^2 R_n}
  \end{split}
  \]
Therefore,
  \[ \int_0^1{e^{\theta \lambda M_n(t)}\,\mathrm{d}t} \leq \frac{4\pi}{1-\theta} e^{\frac{1}{2}\lambda^2 R_n + \ln{n}}.\]
Have $\lambda = 2\sqrt{\frac{\ln{n}}{R_n}}$ and multiply both sides by $n^{-4-\eta}$, where $\eta > 0$. Then
  \[ \int_0^1{e^{\theta \lambda M_n(t)-(4+\eta)\ln{n}}\,\mathrm{d}t} \leq \frac{4\pi}{1-\theta} n^{-(1+\eta)}.\]
The sum over all $n$ converges:
  \[ \sum_{n=1}^{\infty}{\int_0^1{e^{\theta \lambda M_n(t)-(4+\eta)\ln{n}}\,\mathrm{d}t}} \leq \sum_{n=1}^{\infty}{\frac{4\pi}{1-\theta} n^{-(1+\eta)}} < \infty.\]
Since the exponent function is non-negative and the whole sum converges, it is safe to interchange sum and integral:
  \[ \int_0^1{\sum_{n=1}^{\infty}{e^{\theta \lambda M_n(t)-(4+\eta)\ln{n}}}\,\mathrm{d}t} < \infty.\]
Therefore, for almost all $t$ 
  \[\sum_{n=1}^\infty{e^{\theta \lambda M_n(t)-(4+\eta)\ln{n}}}<\infty .\]
Hence, for almost all $t$ there exists $n_0$ such that for all $n \geq n_0$
  \[\theta \lambda M_n(t)-(4+\eta)\ln{n} < 0.\]
It follows that
  \[\lim_{n\rightarrow \infty}{\Pr \left[M_n(t) < \frac{(4+\eta)}{2\theta}\sqrt{R_n\ln{n}}\right]} = 1.\]
\end{proof}

\section{Conclusion}

We have proven that the entangled value of almost any $n$-player symmetric XOR game is $\Theta\left( \frac{\sqrt{\ln{n}}}{n^{1/4}}\right )$ and therefore by a factor of $\sqrt{\ln{n}}$ greater than its classical value. However, our numerical experiments indicate that neither of the coefficients $C_1$ and $2C_2$ in eq. \ref{eq:bounds} are tight.
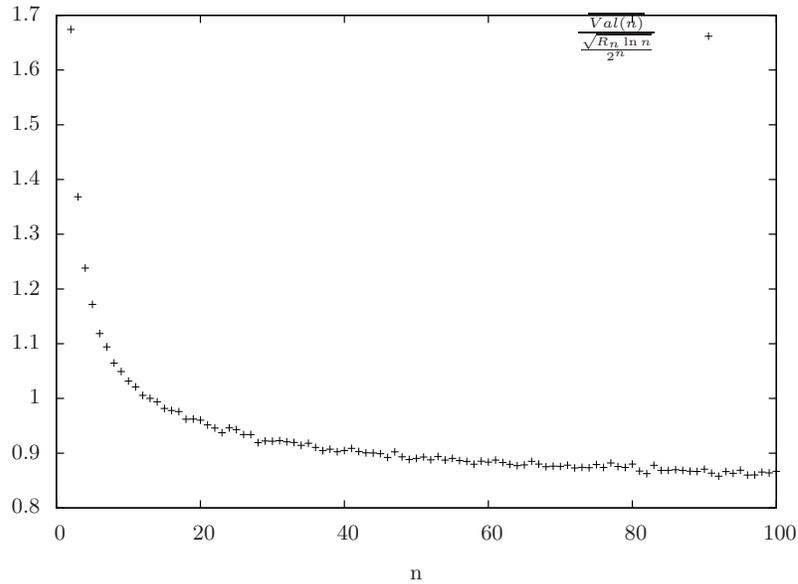
\begin{figure}
\centering
\resizebox{0.95\textwidth}{!}{\input{value}}
\caption{$\frac{\overline{Val(n)}}{\frac{\sqrt{R_n \ln{n}}}{2^n}}$ for a random sample of $n$ player games}
\label{fig:val}
\end{figure}
In Fig. \ref{fig:val} there is plotted the mean value of the coefficient over a sample of $10^5$ games for each $n$ up to 100. We speculate that the actual constant is approaching $\approx 0.85\ldots$.

In this paper we have dealt with a small portion of non-local games. In particular, the case of random non-symmetric games is still open and there has been little progress in multiplayer XOR games with $m-ary$ input. The primary hurdle in the $n$-player $m$-input setting is at the moment it lacks a description in terms of algebraic and analytic expressions. Recently an approach using the theory of operator norms has been successful in proving the entangled value of 3-player, $m$-input XOR game \cite{BV12}.

\bibliographystyle{splncs03}

\phantomsection
\addcontentsline{toc}{chapter}{References}
\bibliography{quantum}

\end{document}

%% file: value.tex
\begingroup
  \makeatletter
  \providecommand\color[2][]{%
    \GenericError{(gnuplot) \space\space\space\@spaces}{%
      Package color not loaded in conjunction with
      terminal option `colourtext'%
    }{See the gnuplot documentation for explanation.%
    }{Either use 'blacktext' in gnuplot or load the package
      color.sty in LaTeX.}%
    \renewcommand\color[2][]{}%
  }%
  \providecommand\includegraphics[2][]{%
    \GenericError{(gnuplot) \space\space\space\@spaces}{%
      Package graphicx or graphics not loaded%
    }{See the gnuplot documentation for explanation.%
    }{The gnuplot epslatex terminal needs graphicx.sty or graphics.sty.}%
    \renewcommand\includegraphics[2][]{}%
  }%
  \providecommand\rotatebox[2]{#2}%
  \@ifundefined{ifGPcolor}{%
    \newif\ifGPcolor
    \GPcolorfalse
  }{}%
  \@ifundefined{ifGPblacktext}{%
    \newif\ifGPblacktext
    \GPblacktexttrue
  }{}%
  \let\gplgaddtomacro\g@addto@macro
  \gdef\gplbacktext{}%
  \gdef\gplfronttext{}%
  \makeatother
  \ifGPblacktext
    \def\colorrgb#1{}%
    \def\colorgray#1{}%
  \else
    \ifGPcolor
      \def\colorrgb#1{\color[rgb]{#1}}%
      \def\colorgray#1{\color[gray]{#1}}%
      \expandafter\def\csname LTw\endcsname{\color{white}}%
      \expandafter\def\csname LTb\endcsname{\color{black}}%
      \expandafter\def\csname LTa\endcsname{\color{black}}%
      \expandafter\def\csname LT0\endcsname{\color[rgb]{1,0,0}}%
      \expandafter\def\csname LT1\endcsname{\color[rgb]{0,1,0}}%
      \expandafter\def\csname LT2\endcsname{\color[rgb]{0,0,1}}%
      \expandafter\def\csname LT3\endcsname{\color[rgb]{1,0,1}}%
      \expandafter\def\csname LT4\endcsname{\color[rgb]{0,1,1}}%
      \expandafter\def\csname LT5\endcsname{\color[rgb]{1,1,0}}%
      \expandafter\def\csname LT6\endcsname{\color[rgb]{0,0,0}}%
      \expandafter\def\csname LT7\endcsname{\color[rgb]{1,0.3,0}}%
      \expandafter\def\csname LT8\endcsname{\color[rgb]{0.5,0.5,0.5}}%
    \else
      \def\colorrgb#1{\color{black}}%
      \def\colorgray#1{\color[gray]{#1}}%
      \expandafter\def\csname LTw\endcsname{\color{white}}%
      \expandafter\def\csname LTb\endcsname{\color{black}}%
      \expandafter\def\csname LTa\endcsname{\color{black}}%
      \expandafter\def\csname LT0\endcsname{\color{black}}%
      \expandafter\def\csname LT1\endcsname{\color{black}}%
      \expandafter\def\csname LT2\endcsname{\color{black}}%
      \expandafter\def\csname LT3\endcsname{\color{black}}%
      \expandafter\def\csname LT4\endcsname{\color{black}}%
      \expandafter\def\csname LT5\endcsname{\color{black}}%
      \expandafter\def\csname LT6\endcsname{\color{black}}%
      \expandafter\def\csname LT7\endcsname{\color{black}}%
      \expandafter\def\csname LT8\endcsname{\color{black}}%
    \fi
  \fi
  \setlength{\unitlength}{0.0500bp}%
  \begin{picture}(7200.00,5040.00)%
    \gplgaddtomacro\gplbacktext{%
      \csname LTb\endcsname%
      \put(726,704){\makebox(0,0)[r]{\strut{} 0.8}}%
      \put(726,1156){\makebox(0,0)[r]{\strut{} 0.9}}%
      \put(726,1609){\makebox(0,0)[r]{\strut{} 1}}%
      \put(726,2061){\makebox(0,0)[r]{\strut{} 1.1}}%
      \put(726,2513){\makebox(0,0)[r]{\strut{} 1.2}}%
      \put(726,2966){\makebox(0,0)[r]{\strut{} 1.3}}%
      \put(726,3418){\makebox(0,0)[r]{\strut{} 1.4}}%
      \put(726,3870){\makebox(0,0)[r]{\strut{} 1.5}}%
      \put(726,4323){\makebox(0,0)[r]{\strut{} 1.6}}%
      \put(726,4775){\makebox(0,0)[r]{\strut{} 1.7}}%
      \put(858,484){\makebox(0,0){\strut{} 0}}%
      \put(2047,484){\makebox(0,0){\strut{} 20}}%
      \put(3236,484){\makebox(0,0){\strut{} 40}}%
      \put(4425,484){\makebox(0,0){\strut{} 60}}%
      \put(5614,484){\makebox(0,0){\strut{} 80}}%
      \put(6803,484){\makebox(0,0){\strut{} 100}}%
      \put(3830,154){\makebox(0,0){\strut{}n}}%
    }%
    \gplgaddtomacro\gplfronttext{%
      \csname LTb\endcsname%
      \put(5816,4602){\makebox(0,0)[r]{\strut{}$\frac{\overline{Val(n)}}{\frac{\sqrt{R_n \ln{n}}}{2^n}}$}}%
    }%
    \gplbacktext
    \put(0,0){\includegraphics{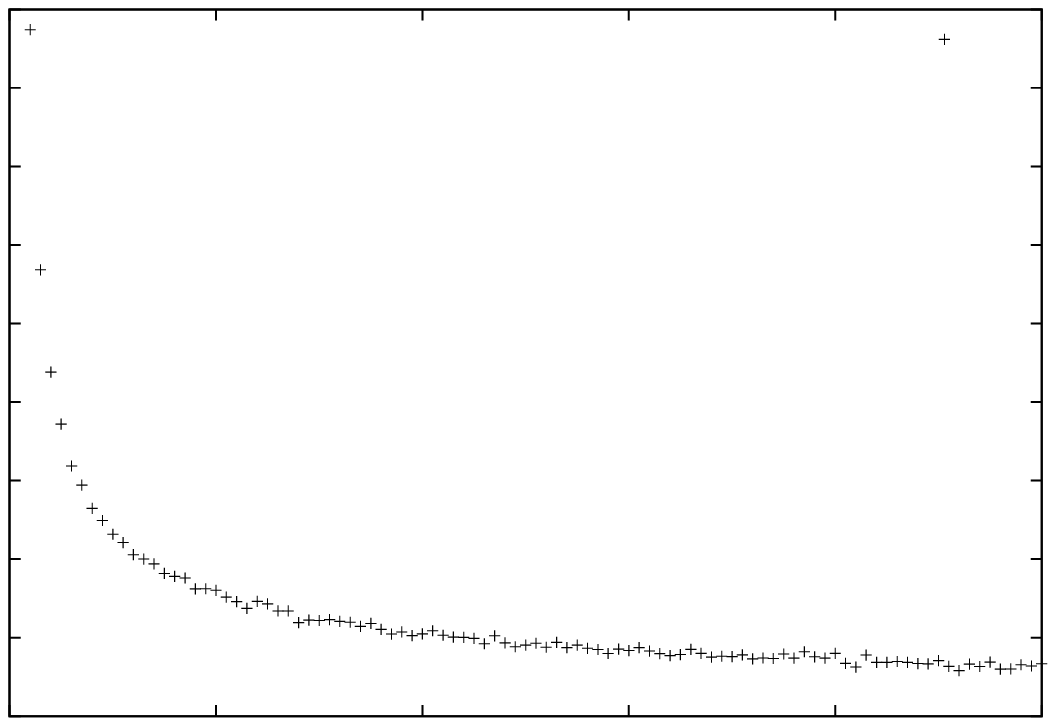}}%
    \gplfronttext
  \end{picture}%
\endgroup